\documentclass[review]{elsarticle}

\usepackage{lineno,hyperref, caption}
\usepackage{algorithm}
\usepackage{algpseudocode}
\usepackage{algorithmicx}
\usepackage{wrapfig}
\usepackage{amsmath,amssymb}
\usepackage{rotating}
\modulolinenumbers[5]

\journal{Arxiv}

\newcommand\hmatch{H}
\newcommand\vmatch{V}
\newcommand\xmatch{X}
\newcommand\vmove{\updownarrow}
\newcommand\hmove{\leftrightarrow}
\newcommand{\dmove}{
	\begin{turn}{45}
		\raisebox{-1ex}{$\leftrightarrow$}
	\end{turn}
}
\newcommand{\degree}{^\circ}
\newdefinition{definition}{Definition}
\newtheorem{theorem}{Theorem}
\newtheorem{lemma}[theorem]{Lemma}
\newdefinition{remark}{Remark}
\newdefinition{corollary}{Corollary}
\newproof{proof}{Proof}
\newproof{pot}{Proof of Theorem \ref{thm2}}
\newcommand{\Mod}[1]{\ (\mathrm{mod}\ #1)}

\begin{document}
\begin{frontmatter}
\title{Taming the Knight's Tour: \\Minimizing Turns and Crossings}

%% Group authors per affiliation
\author[1]{Juan Jose Besa\fnref{fn1}}
\ead{jjbesavi@uci.edu}

\author[1]{Timothy Johnson\fnref{fn1}}
\ead{timjohnson314@gmail.com}

\author[1]{Nil Mamano%
	\fnref{fn1}}
\ead{nmamano@uci.edu}

\author[1]{Martha~C.~Osegueda\corref{cor1}%
\fnref{fn1}}
\ead{mosegued@uci.edu}

\author[2]{Parker Williams}
\ead{paw2140@columbia.edu}

\address[1]{University of California, Irvine, Irvine, CA 92697, United States}
\address[2]{Columbia University, New York, NY 10027, United States}

\cortext[cor1]{Corresponding author}
\fntext[fn1]{The authors were supported by NSF Grant CCF-1616248 and NSF Grant 1815073.}

\begin{abstract}
We introduce two new metrics of ``simplicity'' for knight's tours: the number of turns and the number of crossings. We give a novel algorithm that produces tours with $9.25n+O(1)$ turns and $12n+O(1)$ crossings on an $n\times n$ board, and we show lower bounds of $(6-\epsilon)n$ and $4n-O(1)$ on the respective problems of minimizing these metrics. Hence, our algorithm achieves approximation ratios of $9.25/6+o(1)$ and $3+o(1)$. Our algorithm takes linear time and is fully parallelizable, i.e., the tour can be computed in $O(n^2/p)$ time using $p$ processors in the CREW PRAM model. We generalize our techniques to rectangular boards, high-dimensional boards, symmetric tours, odd boards with a missing corner, and tours for $(1,4)$-leapers. In doing so, we show that these extensions also admit a constant approximation ratio on the minimum number of turns, and on the number of crossings in most cases.
\end{abstract}
\end{frontmatter}

\section{Introduction}
The game of chess is a fruitful source of mathematical puzzles. The puzzles often blend an appealing aesthetic with interesting and deep combinatorial properties~\cite{watkins2012across}.
An old and well-known problem is the knight's tour problem. A \textit{knight's tour} in a generalized $n\times m$ board is a path through all $nm$ cells such that any two consecutive cells are connected by a ``knight move'' (Figure~\ref{fig:knightmove}). For a historic treatment of the problem, see~\cite{rouse1974mathematical}.
A knight's tour is \textit{closed} if the last cell in the path is one knight move away from the first one. Otherwise, it is \textit{open}. This paper focuses solely on closed tours, so henceforth we omit the distinction. The knight's tour problem is a special case of the Hamiltonian cycle problem, in which  we find a simple cycle that visits all the nodes for a specific class of graphs. These graphs are formed by representing each cell on the board as a node and connecting cells a knight move apart. 

Existing work focuses on the questions of existence, counting, and construction algorithms. 

In general, the goal of existing algorithms is to find \textit{any} knight's tour. We propose two new metrics that capture simplicity and structure in a knight's tour. 

We associate each cell in the board with a point $(i,j)$ in the plane, where $i$ represents the row and $j$ represents the column.
\begin{figure}[t!]%[8]{r}{0.25\textwidth}
	\centering
	\includegraphics[width=0.25\linewidth,page=2]{figures}
	\captionof{figure}{A knight moves two units on one axis and one on the other. %(i.e., a knight can move from one cell to another if the distance between their coordinates is $\sqrt{5}$)
	}
	\label{fig:knightmove}
\end{figure}

\newtheorem{problem}{Problem}
\begin{definition}[Turn]
	Given a knight's tour, a \emph{turn} is \emph{a triplet of consecutive cells with non-colinear coordinates.}
\end{definition}

\begin{problem}[Minimum turn knight's tour]\label{pro:turns}
	Given a rectangular $n\times m$ board such that a knight's tour exists, find a knight's tour minimizing the number of turns.
\end{problem}

\begin{definition}[Crossing]
	Given a knight's tour, a \emph{crossing} occurs \emph{when the two line segments corresponding to moves in the tour intersect}. E.g. if $\{c_1,c_2\}$ and $\{c_3,c_4\}$ are two distinct pairs of consecutive cells visited along the tour, a crossing happens if the open line segments $(c_1,c_2)$ and $(c_3,c_4)$ intersect.
\end{definition}

\begin{problem}[Minimum crossing knight's tour]\label{pro:crossings}
	Given a rectangular $n\times m$ board such that a knight's tour exists, find a knight's tour minimizing the number of crossings.
\end{problem}

Knight's tours are typically visualized by connecting consecutive cells by a line segment. Turns and crossings make the sequence harder to follow. Minimizing crossings is a central problem in \emph{graph drawing}, the sub-field of graph theory concerned with the intelligible visualization of graphs (e.g., see the survey in~\cite{herman2000graph}). Problem~\ref{pro:crossings} is the natural adaptation for knight's tours. Minimizing turns has many applications in computational geometry~\cite{Aggarwal97theangular,arkin2003,arkin2001}, where the number of turns in a polygonal line is sometimes referred to as \textit{link length}~\cite{meertens1994,COLLINS1998317}. Problem~\ref{pro:turns} asks for the (self-intersecting) polygon with the smallest number of vertices that represents a valid knight's tour. 

\subsection{Our contributions.}

We propose a novel algorithm for finding knight's tours with the following features.
\begin{itemize}
	\item $9.25n+O(1)$ turns and $12n+O(1)$ crossings on an $n\times n$ board.
	\item A $9.25/6+o(1)$ approximation factor on the minimum number of turns (Problem~\ref{pro:turns}).\footnote{By $o(1)$, we mean a function $f(n)$ such that for any constant $\epsilon>0$, there is an $n_0$ such that, for all $n\geq n_0$, $f(n)<\epsilon$.} The approximation factor is achieved by showing a $(6-\epsilon)n$ lower bound, for any $\epsilon>0$, on the number of turns of any knight's tour (Theorem~\ref{thm:lowerbound}).
	\item A $3+o(1)$ approximation factor on the minimum number of crossings (Problem~\ref{pro:crossings}). The approximation factor is achieved by showing a $4n-O(1)$ lower bound on the number of crossings of any knight's tour (Theorem~\ref{thm:lowerboundcrossings}).
	\item A $O(nm)$ run-time on an $n\times m$ board, i.e., linear on the number of cells, which is optimal.
	\item The algorithm is fully parallelizable: the tour can be computed in $O(nm/p)$ time with $O(p)$ processors in the CREW PRAM model (Theorem~\ref{thm:parallel}).
	\item It can be generalized to most typical variations of the problem: higher-dimensional cubical boards, rotationally symmetric tours, tours in boards with odd width and height that skip a corner cell, and tours for $(1,4)$-leapers, also known as \textit{giraffes}, which move one cell in one dimension and four in the other.
	\item The algorithm can be simulated by hand with ease. This is of particular interest in the context of recreational mathematics and mathematics outreach. %The knight's tour problem has been used to introduce Hamiltonian cycles in discrete math courses~\cite{schwenk1991rectangular}, so .
\end{itemize}

The paper is organized as follows. Section~\ref{sec:related} gives an overview of the literature on the knight's tour problem and its variants. Section~\ref{sec:methods} describes the algorithm and its complexity analysis. We prove the approximation ratios in Section~\ref{sec:lower}. Section~\ref{a:extensions} shows the mentioned variations. We conclude in Section~\ref{sec:conclusions}. The tours produced by the algorithm can be visualized interactively for different board dimensions at~\url{http://nmamano.com/MinCrossingsKnightsTour}.

\section{Extended Related Work}\label{sec:related}
Despite being over a thousand years old~\cite{watkins2012across}, the knight's tour problem is still an active area of research. We review the key questions considered in the literature. 

\textbf{Existence.}
In rectangular boards, a tour exists as long as one dimension is even and the board size is large enough; no knight's tour exists for dimensions $1\times n$, $2\times n$ or $4\times n$, for any $n\geq 1$ and, additionally, none exist for dimensions $3\times 6$ or $3\times 8$~\cite{schwenk1991rectangular}. In three dimensions or higher, the situation is similar: a tour exists only if at least one dimension is even and the board is large enough~\cite{demaio2007chessboards,demaio2011prism,erde2012closed}. In the case of open knight's tours, a tour exists in two dimensions if both dimensions are at least $5$~\cite{cull1978knight,conrad1994solution}.

\textbf{Counting.}
 The number of closed knight's tours in an even-sized $n\times n$ board is at least $\Omega(1.35^{n^2})$ and at most $4^{n^2}$~\cite{kyek1997bounds}. The exact number of knight's tours in the standard $8\times 8$ board is $26,534,728,821,064$~\cite{mckay1997knight}. Algorithms for enumerating multiple~\cite{shufelt1993generating} or all~\cite{Alwan1992reentrant} knight's tours have also been studied.

\textbf{Algorithms.}
Historically, greedy algorithms have been popular for the knight's tour problem. The idea is to construct the tour in order, one step at a time, according to some heuristic. Warnsdorff's rule and its refinements~\cite{pohl1967method,Alwan1992reentrant,squirrel1996warnsdorff} work well in practice for small boards, but do not scale to larger boards~\cite{parberry1996scalability}. The basic idea is to choose the next node with fewest possible continuations, which is also useful for the Hamiltonian cycle problem~\cite{pohl1967method}.

To our knowledge, all efficient algorithms for arbitrary board sizes before this paper are based on a divide-and-conquer approach.
The tour is solved for a finite set of small, constant-size boards. Then, the board is covered by these smaller tours like a mosaic. The small tours are connected into a single one by swapping a few carefully chosen knight moves. This can be done in a bottom-up~\cite{schwenk1991rectangular,conrad1994solution,demaio2007chessboards,demaio2011prism,kamcev2014generalised} or a top-down recursive~\cite{parberry1997efficient,lin2005optimal} fashion. These processes are simple and can be done in time linear on the number of cells and, like ours, are also highly parallelizable~\cite{conrad1994solution,parberry1997efficient} since they are made of repeating patterns.

Divide-and-conquer is not suitable for finding tours with a small number of turns or crossings. Since each base solution has constant size, an $n\times n$ board is covered by $\Theta(n^2)$ of them, and each one contains turns and crossings. Thus, the divide-and-conquer approach necessarily results in $\Theta(n^2)$ turns and crossings. In contrast, our algorithm has $O(n)$ of each.

\textbf{Extensions.}
The above questions have been considered in related settings. Extensions can be classified into three categories, which may overlap:
\begin{itemize}
	\item \textbf{Tours with special properties.} Our work can be seen as searching for tours with special properties. \textit{Magic knight's tours} are also in this category: tours such that the indices of each cell in the tour form a magic square (see~\cite{beasley2012magic} for a survey). %Similarly, in \textit{figured knight's tours}, cells with arithmetically related indices are arranged in geometric patterns~\cite{jelliss1997figured}. %For example, a tour where the eight cells with perfect square indices in an $8\times 8$ board lie in the same row.	
	
	The study of symmetry in knight's tours dates back at least to 1917~\cite{bergholt2001memoirs}. Symmetric tours under 90 degree rotations exist in $n\times n$ tours where $n\geq 6$ and $n$ is of the form $4k+2$ for some $k$~\cite{dejter1983symmetry}. Parberry extended the divide-and-conquer approach to produce tours symmetric under 90 degree rotations~\cite{parberry1997efficient}. Jelliss provided results on which rectangular board sizes can have which kinds of symmetry~\cite{jelliss1999symmetry}.
	
	Both of our proposed problems are new, but minimizing crossings is related to the \textit{uncrossed knight's tour problem}, which asks to find the longest sequence of knight moves without \textit{any} crossings~\cite{yarbrough1969uncrossed,jelliss1999intersecting,fischer2006new,kumar2008noncrossing}. This strict constraint results in incomplete tours.
	
\item \textbf{Board variations.} Besides higher dimensions, knight's tours have been considered in other boards, such as torus boards, where the top and bottom edges are connected, and the left and right edges are also connected. Any rectangular torus board has a closed tour~\cite{watkins1997torus}.
Another option is to consider boards with odd width and height. Since boards with an odd number of cells do not have tours, it is common to search for tours that skip a specific cell, such as a corner cell~\cite{parberry1997efficient}.

\item \textbf{Move variations.} An $(i,j)$\textit{-leaper} is a generalized knight that moves $i$ cells in one dimension and $j$ in the other~\cite{nash1959abelian}. The knight is a $(1,2)$-leaper. Knuth studied the existence of tours for general $(i,j)$-leapers in rectangular boards~\cite{knuth1994leaper}. 
	Tours for \textit{giraffes} ($(4,1)$-leapers) were provided in~\cite{dejter1983symmetry}
	 using a divide-and-conquer approach. Chia and Ong~\cite{chia2005generalized} studied which board sizes admit generalized $(a,b)$-leaper tours. Kamčev~\cite{kamcev2014generalised} showed that any board with sufficiently large and even size admits $(2,3)$\mbox{-,} $(2,5)$\mbox{-,} and $(a,1)$-leaper tours for any even $a$, and generalized this to any higher dimensions. Note that $a$ and $b$ are required to be coprime and not both odd, or no tour can exist~\cite{kamcev2014generalised}.
\end{itemize} 

\begin{figure}[b]
	\centering
	\includegraphics[width=0.99\linewidth,page=3]{figures}
	\caption{Quartet of knights moving in unison leaving no unvisited squares. In a straight move, the starting and ending position of the quartet overlap because two knights remain in place.}
	\label{fig:quartet}
\end{figure}

\section{The Algorithm}\label{sec:methods}
Given that one of the dimensions must be even for a tour to exist, we assume, without loss of generality, that the width $w$ of the board is even, while the height $h$ can be even or odd. We also assume that $w\geq 16$ and $h\geq 12$. The construction still works for some smaller sizes, but may require tweaks to its most general form described here.

\textbf{Quartet moves.}
What makes the knight's tour problem challenging is that knight jumps leave ``gaps''. Our first crucial observation is that a quartet of four knights arranged in a square $2\times 2$ formation can move ``like a king'': they can move horizontally, vertically, or diagonally without leaving any gaps (Figure~\ref{fig:quartet}).

By using the ``formation moves'' depicted in Figure~\ref{fig:quartet}, four knights can easily cover the board moving vertically and horizontally while remaining in formation. Of course, we need to traverse the entire board in a single cycle, not four paths. Thihs is solved with special structures placed in the bottom-left and top-right corners of the board, which we call \textit{junctions}, and which tie the four paths together to create a single cycle. Using only straight formation moves leads to tours with a large number of turns and crossings. Fortunately, two consecutive diagonal moves in the same direction introduces no turns or crossings, so our main idea is to use as many diagonal moves as possible. 
This leads to the general pattern in Figure~\ref{fig:mainconstr}.

\begin{figure}
	\centering
	\includegraphics[width=1\linewidth,page=4]{figures}
	\caption{The sequence of quartet moves \textit{(left)} and the resulting knight's tour \textit{(right)} in a $30\times 30$ board. Arrows illustrate sequences of repeated formation moves. Starting from the bottom-left square of the board, the single knight's tour follows the colored sections in order red, green, light red, purple, blue, orange, light blue, black, and back to red.}
	\label{fig:mainconstr}
\end{figure}

The full algorithm is Algorithm~\ref{alg:main}.
The formation starts at a junction at the bottom-left corner and ends at a junction at the top-right corner. To get from one to the other, it zigzags along an odd number of parallel diagonals, alternating between downward-right and upward-left directions. 
The junctions in Figure~\ref{fig:junctions} have a \textit{height}, which influences the number of diagonals traversed by the formation.

At the bottom-left corner, we use a junction with height 5. At the top-right corner, we use a junction with height between 5 and 8. Choosing the height as in Algorithm~\ref{alg:main} guarantees that, for any board dimensions, an odd number of diagonals fit between the two junctions.
Sequence $1$ in Figure~\ref{fig:othercorners}, which we call the \textit{heel}, is used to transition between diagonals along the top and bottom edges of the board.
The two non-junction corners require special sequences of quartet moves, as depicted in Figure~\ref{fig:othercorners}. Sequences $1,2,3,$ and $0$ are used when the last heel ends $0,2,4,$ and $6$ columns away from the vertical edge, respectively. These variations and the top-right junction are \textit{predictable} because they cycle as the board dimensions grow, so in Algorithm~\ref{alg:main} we give expressions for them in terms of $w$ and $h$.

\begin{algorithm}[tp!]
	\caption{Knight's tour algorithm for even width $w\geq 16$ and height $h\geq 12$.}
	\label{alg:main}
	\begin{algorithmic}
		\State 1. Fill the corners of the board as follows:
		\begin{description}
			\item[Bottom-left:] first junction in Figure~\ref{fig:junctions}.
			\item[Top-right:] junction of height $5+((w/2+h-1)\mbox{ mod }4)$ in Figure~\ref{fig:junctions} except the first one.
			\item[Bottom-right:] Sequence $(w/2+2)\mbox{ mod }4$ in Figure~\ref{fig:othercorners}.
			\item[Top-left:] Sequence $(3-h)\mbox{ mod }4$ in Figure~\ref{fig:othercorners} rotated 180 degrees.
		\end{description}
		\State 2. Connect the four corners using formation moves, by moving along diagonals from the bottom-left corner to the top-right corner as in Figure~\ref{fig:mainconstr}. To transition between diagonals:
		\begin{description}
			\item[Vertical edges:] use a double straight up move  (Figure~\ref{fig:quartet}).
			\item[Horizontal edges:] use Sequence 1 in Figure~\ref{fig:othercorners}.
		\end{description}
	\end{algorithmic}
\end{algorithm}

% It is important to keep these figures next to the Algorithm definition
% becaues they are repeatedly refered to from the algorithm.
\begin{figure}[h!]
	\centering
	\includegraphics[width=0.99\linewidth,page=5]{figures}
	\caption{Junctions used in our construction.}
	\label{fig:junctions}
\end{figure}

\begin{figure}[h!]
	\centering
	\includegraphics[width=.8\linewidth,page=7]{figures}
	\caption{The four possible cases for the bottom-right corner.}
	\label{fig:othercorners}
\end{figure}

Figure~\ref{fig:sizes} shows the manifestation of these variations in different board dimensions. The sequence of quartet moves is easy to construct by hand. One starts from the bottom-left junction, which is always the same. One can then proceed in order. The other corners can be figured out as one encounters them. For the non-junction ones, one need only keep in mind that if the usual heel does not reach the right edge exactly, at least four columns are needed for the next transition. For instance, in Sequence $2$ in Figure~\ref{fig:othercorners}, the last heel is cut short to leave four columns instead of two. At the final corner, simply stop when there are between 5 and 8 rows left along the right edge, and plug in the corner of the appropriate size.

\begin{figure}
	\centering
	\includegraphics[width=0.99\linewidth,page=8]{figures}
	\caption{Traversal of the board by the knight quartet for different board dimensions.}
	\label{fig:sizes}
\end{figure}

\subsection{Correctness}\label{sec:correctness}
It is clear that the construction visits every cell, and that every node in the underlying graph of knight moves has degree two. However, it remains to be argued that the construction is actually a single closed cycle for any board dimensions. For this, we need to consider the choice of junctions.

A junction is a pair of disjoint knight paths whose four endpoints are adjacent as in the quartet formation. Thus, the paths in the bottom-left junction connect the four knights in two pairs. Denote the four knight positions in the formation by $tl, tr, bl, br$, where the first letter indicates top/bottom and the second left/right. We consider the three possible \textit{positional matchings} with respect to these positions: horizontal matching $\hmatch=(tl,tr), (bl,br)$, vertical matching $\vmatch=(tl,bl), (tr,br)$, and cross matching $\xmatch=(tl,br), (tr,bl)$. Let $\mathcal{M} = \{\hmatch, \vmatch, \xmatch\}$ denote the set of positional matchings.
We are interested in the effect of executing formation moves on the positional matching. A formation move does not change which knights are matched with which, but a non-diagonal move changes their positions, and thus their matchings. %labels $tl, tr, bl, br$ also change. 
For instance, a horizontal matching becomes a cross matching after a straight move to the right. 

\begin{table}[b]
	\centering
	\begin{tabular}{c|cccccc}
		& $\dmove$ & $\vmove$ & $\hmove$ & $\vmove\hmove$ & $\hmove\vmove$ & $\vmove\hmove\vmove$ \\ \hline
		$\vmatch$	& $\vmatch$ & $\xmatch$ & $\vmatch$ & $\hmatch$ & $\xmatch$ & $\hmatch$ \\
		$\hmatch$		& $\hmatch$ & $\hmatch$ & $\xmatch$ & $\xmatch$ & $\vmatch$ & $\vmatch$ \\
		$\xmatch$	& $\xmatch$ & $\vmatch$ & $\hmatch$ & $\vmatch$ & $\hmatch$ & $\xmatch$
	\end{tabular}
	\caption {Result of applying each type of formation move, as well as three compositions of sequences of moves, to each formation matching.} \label{tab:perms}
\end{table}

\begin{table}[b]
	\centering
	\begin{tabular}{c|cccccc}
		& $\dmove$ & $\vmove$ & $\hmove$ & $\vmove\hmove$ & $\hmove\vmove$ & $\vmove\hmove\vmove$ \\ \hline
		$\dmove$	            & $\dmove$ & $\vmove$ & $\hmove$ & $\vmove\hmove$ & $\hmove\vmove$ & $\vmove\hmove\vmove$  \\
		$\vmove$	            & $\vmove$ & $\dmove$ & $\vmove\hmove$ & $\hmove$ & $\vmove\hmove\vmove$ & $\hmove\vmove$ \\
		$\hmove$	            & $\hmove$ & $\hmove\vmove$ & $\dmove$ & $\vmove\hmove\vmove$ & $\vmove$ & $\vmove\hmove$ \\
		$\vmove\hmove$	        & $\vmove\hmove$ & $\vmove\hmove\vmove$ & $\vmove$ & $\hmove\vmove$ & $\dmove$ & $\hmove$ \\
		$\hmove\vmove$	        & $\hmove\vmove$ & $\hmove$ & $\vmove\hmove\vmove$ & $\dmove$ & $\vmove\hmove$ & $\vmove$ \\
		$\vmove\hmove\vmove$	& $\vmove\hmove\vmove$ & $\vmove\hmove$ & $\hmove\vmove$ & $\vmove$ & $\hmove$ & $\dmove$
	\end{tabular}
	\caption {Cayley table for the group of positional matching permutations.}\label{tab:cayley}
\end{table}

It is easy to see that a straight move upwards or downwards has the same effect on the positional matching. Similarly for left and right straight moves. Thus, we classify the formation moves in Figure~\ref{fig:quartet} (excluding double straight moves, which are a composition of two straight moves) into vertical straight moves $\vmove$, horizontal straight moves $\hmove$, and diagonal moves $\dmove$. Let $\mathcal{S}=\{\vmove,\hmove,\dmove\}$ denote the three types of quartet moves. We see each move type $s\in\mathcal{S}$ as a function $s:\mathcal{M}\to\mathcal{M}$. The move types in $\mathcal{S}$, seen as functions are, in fact, permutations (see first three columns in Table~\ref{tab:perms}). For example, the diagonal move $\dmove$ is just the identity.
It follows that \textit{any} sequence of formation moves permutes the positional matchings, according to the composed permutation of each move in the sequence. The composed permutation of a sequence of moves $S=(s_1,\ldots,s_k)$, where each $s_i\in\mathcal{S}$, on a given positional matching $M\in\mathcal{M}$ is $S(M)=s_1\circ\cdots \circ s_k(M)$. There are six possible permutations of the three positional matchings, three of which correspond to the ``atomic'' formation moves $\dmove,\vmove,$ and $\hmove$. The other three permutations can be obtained by composing atomic moves, for instance, with the compositions $\vmove\hmove, \hmove\vmove,$ and $\vmove\hmove\vmove$ (Table~\ref{tab:perms}). Thus, any sequence of moves permutes the positional matchings in the same way as one of the sequences in the set $\{\dmove, \vmove, \hmove, \vmove\hmove, \hmove\vmove, \vmove\hmove\vmove\}$. This is equivalent to saying that this set, under the composition operation, is isomorphic to the symmetric group of degree three. 
Table~\ref{tab:cayley} shows the Cayley table of this group.

Let $T_{w,h}$ be the sequence of formation moves that goes from the bottom-left junction to the top-right one in Algorithm~\ref{alg:main} in a $w\times h$ board.

\begin{lemma}\label{lem:seq}
For any even $w\geq 16$ and any $h\geq 12$, $T_{w,h}(\hmatch)=\hmatch$.
\end{lemma}

\begin{proof}	
	We show that the entire sequence of moves $T_{w,h}$ is either neutral or equivalent to a single vertical move, depending on the board dimensions. According to Table~\ref{tab:perms}, this suffices to prove the lemma. 
	
	The sequence $T_{w,h}$ consists mostly of diagonal moves, which are neutral. The transition between diagonals along the vertical edges consist of two vertical moves, which are also neutral ($\vmove\vmove=\dmove$). The heel is also neutral, as it consists of the sequence $\vmove\vmove\hmove\hmove\vmove\hmove\hmove\vmove$ (omitting diagonal moves) which is again equivalent to $\dmove$. This is easily seen by noting that any two consecutive vertical or horizontal moves cancel out. For instance, Figure~\ref{fig:neutralheel} shows that if the positional formation is $H$ when entering the heel, then it is also $H$ when leaving the heel.
	Thus, $T_{w,h}$ can be reduced to the composition of the sequences in the bottom-right and top-left corners shown in Figure~\ref{fig:othercorners}. As mentioned, Sequence 1, the heel, is neutral. It is easy to see that the other sequences (counting each part of Sequence 2 separately) is equivalent to $\vmove$. Thus, we get that $T_{w,h}$ is simply the composition of zero to four vertical moves, depending on the width and height of the board. Since two consecutive vertical moves cancel out, this further simplifies to zero or one vertical moves.
\end{proof}

\begin{figure}
	\centering
	\includegraphics[width=.7\linewidth,page=9]{figures}
	\caption{Formation moves \textit{(left)} and individual knight moves \textit{(right)} in the heel.}
	\label{fig:neutralheel}
\end{figure}

\begin{theorem}[Correctness]\label{thm:correctness}
	Algorithm~\ref{alg:main} outputs a valid knight's tour in any board with even width $w\geq 16$ and with height $h\geq 12$.
\end{theorem}
\begin{proof}
	Clearly, the formation moves in our construction yield a set of disjoint cycles in the underlying knight-move graph. We prove that they actually form a single cycle.
	Given a set of disjoint cycles in a graph, \textit{contracting} a node in one of the cycles is the process of removing it and connecting its two neighbors in the cycle. Contracting a node in a cycle of length $\geq 3$ does not change the number of cycles.
	Thus, consider the remaining graph if we contract all the nodes except the four endpoints of the top-right junction. 
	
	Since we use a vertical matching in the top-right junction, contracting the non-endpoint nodes inside the top-right junction leaves the two edges corresponding to the vertical matching. In the bottom-left junction we use a horizontal matching so, by Lemma~\ref{lem:seq}, contracting the nodes outside the top-right junction leaves the edges corresponding to a horizontal matching. Thus, the resulting graph is a single cycle of four nodes.
\end{proof}

The choice of matchings at the junctions is important; using a horizontal matching in the top-right junction would not result in a knight's tour.
On the other hand, the \textit{shapes} of the junctions are not important. We use the junctions in Figure~\ref{fig:junctions} because they fit nicely with the rest of the construction, but alternative shapes are possible. As an alternative, Figure~\ref{fig:altjunctions} shows a $6\times 6$ junction which can replace all the junctions in Figure~\ref{fig:junctions}.

\begin{figure}
	\centering
	\includegraphics[width=0.75\linewidth,page=6]{figures}
	\caption{\textbf{Left:} alternative bottom-left junction and its reflection over the $x=y$ axis. This junction and its reflection alone can be used instead of the junctions in Figure~\ref{fig:junctions}. Because the ``exit point'' of the junction is split evenly by the $x=y$ line, it can be reflected over that line to produce a vertical matching instead of a horizontal one. \textbf{Right:} instead of having a junction for each specific height, we can influence the number of diagonals made by the quartet formation with the first quartet move out of the junction. This requires some special quartet moves to cover the space to the side and above the junction. The figure illustrates one of the three cases.}
	\label{fig:altjunctions}
\end{figure}

\subsection{Parallel Implementation}
	 Algorithms~\ref{alg:main} only uses the width and the height of the board to determine the configuration on the corners. Consequently, we can determine the two knight moves incident on each cell in constant time: there is a constant number of cells near the corners; the moves in the first and last two columns are based on repeating vertical formation moves; the moves in the first and last four rows are based on the repeating heel pattern; and the remaining moves are based on the diagonal formation moves.
	 
	 Besides knowing the two knight moves incident on each cell, we can also find, in constant time, where the knight will be after a certain number of moves.
	  
\begin{lemma}\label{lem:parallel}
	Let $T$ be the tour produced by Algorithm~\ref{alg:main} on a board with even width $w\geq 16$ and height $h\geq 12$. Assuming that the knight starts at the bottom-left cell $(1,1)$, we can, for any $i$ between $1$ and $w\cdot h$, the $i$-th cell in $T$ using a constant number of arithmetic operations.
\end{lemma}

\begin{proof}
	For concreteness, we assume that $w \equiv 6 \Mod{8}$ and $h \equiv 2\Mod{4}$. According to Algorithm~\ref{alg:main}, this means that the bottom-right and top-left edges have a heel shape.
	This is one of a constant number of cases regarding the values of $w\mbox{ mod }8$ and $h\mbox{ mod } 4$. The other cases follow a similar argument, but with a slightly different case analysis.
	
	The tour $T$ is divided into 8 sections: the paths of the four formation knights and the 2 paths in each junction (each section has a different color in Figure~\ref{fig:mainconstr}). The path of each formation knight can be further divided into three sections that follow regular patterns (Figure~\ref{fig:parallelsections}): starting from the bottom-left corner, the first section follows formation moves between the left and bottom edges. The second section follows formation moves between the top and bottom edges or between the left and right edges (as in Figure~\ref{fig:mainconstr}), depending on the dimensions of the board. The third section follows formation moves between the top and right edges. The length of each of these sections is an arithmetic (quadratic) function of $w$ and $h$, so we can calculate their lengths and map the index $i$ to one of these sections.

\begin{figure}
	\centering
	\includegraphics[width=.6\linewidth,page=28]{figures}
	\caption{Division of each formation knight's path into three sections. Section $2$ is of type $(a)$ (left-to-right) or $(b)$ (top-to-bottom) depending on the dimensions of the board. Junctions are represented in gray.}
	\label{fig:parallelsections}
\end{figure}

	Once we know which section index $i$ belongs to, we can find the exact cell within the section. For concreteness, we assume that $i$ belongs to the section $S$ after the knight leaves the bottom-left junction for the first time and before the formation reaches the top-left or bottom-right corners (Section 1 in Figure~\ref{fig:parallelsections}). This corresponds to the black path in Figure~\ref{fig:parallel}. Again, the other sections can be analyzed using a similar argument.
	
	The index of the first cell in $S$ is $7$, since indices $1$ to $6$ correspond to cells inside the bottom-left junction (the red path in Figure~\ref{fig:parallel}).
	We call the cells of $S$ in column $1$ \textit{waymarks}. Figure~\ref{fig:parallel} shows that the index increment between two consecutive waymarks is $8$ more than between the previous pair. This is because each waymark is $4$ rows above the previous one, so the path between two waymarks has to go down and up $4$ more rows than the previous pair using diagonal formation moves, which only go up or down one row at a time.
	
	Given an index $i$ in $S$, the number of waymarks before or at index $i$ is $\lfloor \frac{\sqrt{i+2}-1}{2}\rfloor$. Hence, we can find the last waymark $x$ before (or at) index $i$, its row $r_x$, its index $i_x$, and the index offset $k$ since the last waymark, $k = i-i_x$. Given these values, we can find the cell at index $i$. We consider three cases: if $r_x-k\geq 4$ the cell is between the waymark and the heel, and its coordinates are $(r_x-k,1+2k)$; if $4>r_x-k\geq -1$ the cell is one of the five cells in the heel; otherwise, it is between the heel and the next waymark.
	
	For example, if index $85$ is in $S$, then, there are $\lfloor \frac{\sqrt{85+2}-1}{2}\rfloor=4$ waymarks before that index, the last one has index $79$, and its coordinates are $(19,1)$. The offset since the last waymark is $85-79=6$, so the cell at index $85$ is $(19-6, 1+12)=(13,13)$.
\end{proof}

\begin{figure}
	\centering
	\includegraphics[width=.9\linewidth,page=26]{figures}
	\caption{Bottom-left corner of a board showing the indices of the waymarks in Lemma~\ref{lem:parallel} and the increment between their indices in parentheses.}
	\label{fig:parallel}
\end{figure}

\begin{theorem}\label{thm:parallel}
	Algorithm~\ref{alg:main} can be computed in $O(wh/p)$ time using $p$ processors in the CREW PRAM model on a board with even width $w\geq 16$ and height $h\geq 12$, assuming arithmetic operations on numbers of the order of $w$ and $h$ take constant time.
\end{theorem}

\begin{proof}
	Lemma~\ref{lem:parallel} states that each index can be computed independently of the rest, so each processor can compute a section of the tour without any data dependencies between processors.
\end{proof}

%Conversely, given a cell $(r,c)$, with $1\leq r\leq n, 1\leq c\leq m$, we can find the index in the tour where the knight is at cell $(r,c)$ also in constant time.

\section{Approximation Ratios}\label{sec:lower}

In this section, we analyze the approximation ratio that our algorithm achieves for the problems of minimizing turns and crossings. For simplicity, we restrict the analysis to square boards.
First, we briefly discuss the classification of these problems in complexity classes.

\subsection{Computational Complexity}\label{a:comp}
Consider the following decision versions of the problems: is there a knight's tour on an $n\times n$ board with at most $k$ turns (resp. crossings)? 
We do not know if these problems are in $\mathsf{P}$. Furthermore, it may depend on how the input is encoded. Technically, the input consists of two numbers, $n$ and $k$, which can be encoded in $O(\log n+\log k)$ bits. However, it is more natural to do the analysis as a function of the board size (or, equivalently, of the underlying graph on which we are solving the Hamiltonian Cycle problem), that is, $\Theta(n^2)$.
It is plausible that the optimal number of turns (resp. crossings) is a simple arithmetic function of $n$. This would be the case if the optimal tour follows a predictable pattern like our construction, where counting the number of turns or crossings does not require recreating the tour. If this were the case, the problems would be in $\mathsf{P}$ regardless of the input's encoding.

If the input is represented using $\Theta(n^2)$ space, the problems are clearly in $\mathsf{NP}$, as a tour with $k$ turns (resp. crossings) acts as a certificate of polynomial length.

However, unless $\mathsf{P}=\mathsf{NP}$, the problems are not $\mathsf{NP}$-hard:
consider the language $\{1^n01^k\mid \mbox{ a tour exists with at most }k\mbox{ turns in an }n\times n\mbox{ board}\},$
and the analogous language for crossings. These languages are sparse, meaning that, for any given word length, there is a polynomial number of words of that length in the language. Mahaney's theorem states that if a sparse language is $\mathsf{NP}$-complete, then $\mathsf{P}=\mathsf{NP}$~\cite{MAHANEY1982130}. This suggests that the problems are in $\mathsf{P}$, though technically they could also be $\mathsf{NP}$-intermediate.

If the input is represented using $O(\log n+\log k)$ bits, then the problems are in $\mathsf{NEXP}$ because the versions above where the input takes $\Theta(n^2)$ bits are in $\mathsf{NP}$. In this setting, simply listing a tour would require time exponential on the input size.

\subsection{Upper Bounds}\label{sec:upperbounds}
We use our algorithm to find upper bounds on the optimal solutions to Problem~\ref{pro:turns} and Problem~\ref{pro:crossings}. All the turns and crossings in our construction happen near the edges. The four corners account for a constant number of each.
The left and right edges have 8 turns and 10 crossings for each four rows.
The heel from Figure~\ref{fig:neutralheel} has 22 turns and 32 crossings, so the top and bottom edges have that many turns and crossings for each eight columns.

We can reduce the number of turns and crossings slightly by replacing each heel with an alternative set of knight moves that cover the same cells and start and end with the four knights in the same place. In other words, we allow the four knights to temporarily ``break formation'' when they reach the heel and reassemble before leaving it in the next diagonal.

We did an exhaustive search for four knights that cover the squares shown in Figure~\ref{fig:optimalheel}, starting and ending in the two groups of four cells on the top-left. Of all the possible solutions, Figure~\ref{fig:optimalheel} shows two: one minimizing the number of turns and one minimizing the number of crossings. There is no single solution minimizing both. Turns and crossings caused by the path continuations outside the heel are also counted. In our search, the four knights are not required to end in any specific configuration, but the optimal solutions happen to end in the same configuration as the original heel.

\begin{figure}
	\centering
	\includegraphics[width=0.99\linewidth,page=25]{figures}
	\caption{\textbf{Left:} the heel resulting from formation moves. It has 22 turns and 32 crossings. \textbf{Center:} the optimal configuration for minimizing turns. It has 21 turns and 31 crossings. \textbf{Right:} the optimal configuration for minimizing crossings. It has 22 turns and 28 crossings. Crossings are marked with white disks.}
	\label{fig:optimalheel}
\end{figure}

Our construction with the optimized heels has $2\left(\frac{8}{4}+\frac{21}{8}\right)n + O(1)=9.25n+O(1)$ turns or $2\left(\frac{10}{4}+\frac{28}{8}\right)n + O(1)=12n+O(1)$ crossings.

\subsection{Lower Bounds}
We show a lower bound for the optimal solution to each problem.

\subsubsection{Number of Turns}
As a starting point, note that every edge cell (cells in the first or last row or column) \textit{must} contain a turn. This accounts for $4n-4$ turns. A simple argument, sketched in ~\ref{a:easylb}, improves this to a $4.25n-O(1)$ lower bound. Here we focus on the main result, a lower bound of $(6-\epsilon)n$ for any $\epsilon>0$. We start with some intermediate results.

We associate each cell in the board with a point $(i,j)$ in the plane, where $i$ is the row of the cell and $j$ is the column.
An edge cell only has four moves available. We call the directions of these moves $D_1,D_2,D_3,$ and $D_4$, in clockwise order. For an edge cell $c$, let $r_i(c)$, with $1\leq i\leq 4$, denote the ray starting at $c$ and in direction $D_i$. That is, the ray that passes through the cells reachable from $c$ by moving along $D_i$.

Let $a$ and $b$ be two cells along the left edge of the board, with $a$ above $b$. The discussion is symmetric for the other three edges. Given two intersecting rays $r$ and $r'$, one starting from $a$ and one from $b$, let $S(r,r')$ denote the set of cells in the region of the board \textit{bounded} by $r$ and $r'$: the set of cells below or on $r$ and above or on $r'$.
We define the \textit{crown} of $a$ and $b$ as the following set of cells (Figure~\ref{fig:crown}):
$$\mbox{crown}(a,b)=S(r_2(a),r_1(b))\cup S(r_3(a),r_2(b))\cup S(r_4(a),r_3(b)).$$

\begin{figure}
	\centering
	\begin{minipage}{.6\textwidth}
	\centering
\includegraphics[width=0.73\linewidth,page=10]{figures}
\caption{Terminology used in the lower bound. Note that $c$ is a clean cell with respect to the crown of $a$ and $b$ because both of its legs escape the crown.}
\label{fig:crown}
	\end{minipage}%
\hfill
	\begin{minipage}{.35\textwidth}
	\centering
\includegraphics[width=0.5\linewidth,page=11]{figures}
\caption{The black leg would collide with any of the red legs.}
\label{fig:collisions}
	\end{minipage}
\end{figure}

We associate with each edge cell $c$ the two maximal sequences of moves without turns in the tour that have $c$ as an endpoint. We call them the \textit{legs} of $c$. We say that the legs of $c$ begin at $c$ and end at their other endpoint. We say two legs of different edge cells \textit{collide} if they end at the same cell.
Let $C_{a,b}$ denote the set of edge cells along the left edge between $a$ and $b$ ($a$ and $b$ included). The following is easy to see.

\begin{remark}
Any collision between the legs of edge cells in $C_{a,b}$ happens inside $\mbox{crown}(a,b)$.
\end{remark}

We say that a leg of an edge cell in $C_{a,b}$ \textit{escapes} the crown of $a$ and $b$ if it ends outside the crown.
We say an edge cell in $C_{a,b}$ is \textit{clean}, with respect to $C_{a,b}$, if both of its legs escape. %We use the following observation to show that there is only a constant number of clean cells inside a crown.

\begin{lemma}
	Let $m=|C_{a,b}|$ and $k$ be the number of clean cells in $C_{a,b}$.
	The number of turns inside $\mbox{crown}(a,b)$ is \textit{at least} $m+(m-k)/2$.
\end{lemma}

\begin{proof}
Each of the $m$ edge cells is one turn. Each of the $m-k$ non-clean cells have a leg that ends in a turn inside the crown. This turn may be because it collided with the leg of another edge cell in the crown. Thus, there is at least one turn for each two non-clean edge cells.
\end{proof}

\begin{lemma}
Let $a,b$ be two cells along the left edge of the board, with $a$ above $b$.
There are at most 122 clean cells inside $\mbox{crown}(a,b)$.
\end{lemma}

\begin{proof}
First we show that there are at most 60 clean cells such that one of their legs goes in direction $D_1$.
For the sake of contradiction, assume that there are at least 61. Then, there are two, $c$ and $d$, such that $c$ is $60r$ rows above $d$, for some $r\in\mathbb{N},r\geq 1$. The contradiction follows from the fact that the other leg of $c$, which goes along $D_2,D_3,$ or $D_4$, would collide with the leg of $b$ along $D_1$. This is because, for any $l\geq 1$, the leg of $b$ along $D_1$ collides with (Figure~\ref{fig:collisions}):
\begin{itemize}
	\item any leg along $D_2$ starting from a cell $3l$ rows above $b$,
	\item any leg along $D_3$ starting from a cell $5l$ rows above $b$, and
	\item any leg along $D_4$ starting from a cell $4l$ rows above $b$.
\end{itemize}
Since $60r$ is a multiple of $3,4,$ and $5$, no matter what direction the other leg of $c$ goes, it collides with the leg of $d$.
As observed, this collision happens inside the crown. Thus, $c$ and $d$ are not clean.
By a symmetric argument, there are at most 60 clean cells such that one of their legs goes in direction $D_4$.

Finally, note that there can only be two clean cells with legs in $D_2$ \textit{and} $D_3$. By a similar argument, there cannot be two such cells at an even distance of each other; the leg along $D_3$ of the top one would collide with the leg along $D_2$ of the bottom one.  
\end{proof}

\begin{figure}
	\centering
	\begin{minipage}{.48\textwidth}
	\centering
	\includegraphics[width=.99\linewidth,
	page=12]{figures}
	\caption{Each sector of the square shows the process after a different number of iterations: $1,2,3,$ and $4$ iterations on the top, right, bottom, and left sectors, respectively.}
	\label{fig:fractal}
	\end{minipage}%
\hfill
	\begin{minipage}{.48\textwidth}
\centering
	\includegraphics[width=.99\linewidth,
	page=13]{figures}
	\caption{Lower bounds on two ratios. \textbf{Left:} the ratio between the gap between consecutive crowns and the base of the maximum-size crown that fits in the gap is $>0.4$. \textbf{Right:} the ratio between the gap between a crown and a main diagonal and the base of the maximum-size crown that fits in the gap is~$>0.36$.}
	\label{fig:geometry}
	\end{minipage}
\end{figure}

\begin{corollary}\label{cor:122}
Suppose that the crown of $a$ and $b$ has $m\geq 122$ edge cells. Then, there are at least $(m-122)/2$ turns inside the crown at non-edge cells.
\end{corollary}

Consider the iterative process depicted in Figure~\ref{fig:fractal} defined over the unit square. The square is divided in four sectors along its main diagonals. Whereas earlier we used the term `crown' to denote a set of cells, here we use it to denote the polygon with the \textit{shape} of a crown. On the first step, a maximum-size crown is placed on each sector. On step $i>1$, we place $2^{i-1}$ more crowns in each sector. They are maximum-size crowns, subject to being disjoint from previous crowns, in each gap between previous crowns and between the crowns closest to the corners and the main diagonals.

\begin{lemma}\label{lem:fractal}
For any $\epsilon$ with $0<\epsilon<1$, there exists an $i\in\mathbb{N}$ such that at least $(1-\epsilon)$ of the boundary of the unit square is inside a crown after $i$ iterations of the process.
\end{lemma}

\begin{proof}
At each iteration $i>1$, a constant fraction larger than $0.36$ of the length on each side that is not in a crown is added to a new crown (Figure~\ref{fig:geometry}). This gives rise to a series $A_i$ for the fraction of the side inside crowns after $i$ iterations: $A_1=1/3$, $A_{i+1}> A_{i}+0.36(1-A_{i})$ for $i>1$; this series converges to $1$ (we do not prove that $A_1 = 1/3$, but this is inconsequential
because the value of $A_1$ does not affect the convergence of the series).
\end{proof}

\begin{theorem}[Turns lower bound]\label{thm:lowerbound}
	For any constant $\epsilon>0$, there is a sufficiently large $n$ such that any knight's tour on an $n\times n$ board has at least $(6-\epsilon)n$ turns.
\end{theorem}

\begin{proof}
We show a seemingly weaker form of the claim: that there is a sufficiently large $n$ such that any knight's tour on an $n\times n$ board requires $(6-2\epsilon)n-C_\epsilon$ turns, where $C_\epsilon$ is a constant that depends on $\epsilon$ but not on $n$. This weaker form is in fact equivalent because, for sufficiently large $n$, $C_\epsilon<\epsilon n$, and hence $(6-2\epsilon)n-C_\epsilon>(6-3\epsilon)n$. Thus, the claim is equivalent up to a multiplicative factor in $\epsilon$, but note that it is a claim about arbitrarily small $\epsilon$, so it is not affected by a multiplicative factor loss.

Let $i$ be the smallest number of iterations of the iterative process in Figure~\ref{fig:fractal} such that at least $(1-\epsilon)$ of the boundary of the unit square is inside crown shapes. The number $i$ exists by Lemma~\ref{lem:fractal}. Fix $S$ to be the corresponding set of crown shapes, and let $r=|S|$. Note that $r=4(2^i-1)$ is a constant that depends only on $\epsilon$. Now, consider a square board with the crown shapes in $S$ overlaid in top of them. Let the board length $n$ be such that the smallest crown in $S$ contains more than $122$ edge cells.
Then, by Corollary~\ref{cor:122}, adding up the turns at non-edge cells over all the crowns in $S$, we get at least $4n(1-\epsilon)/2-61r$ turns. Adding the $4n-4$ turns at edge cells, we get that the total number of turns is at least $(6-2\epsilon)n-61r-4$. To complete the proof, set $C_\epsilon=61r+4$.
\end{proof}

\begin{corollary}
Algorithm~\ref{alg:main} achieves a $9.25/6+o(1)$ approximation on the minimum number of turns.
\end{corollary}

\begin{proof}
In an $n\times n$ board, let $ALG(n)$ denote the number of turns in the tour produced by Algorithm~\ref{alg:main} (with the optimized heels from  Section~\ref{sec:upperbounds}), and $OPT(n)$ denote the minimum number of turns.
Let $\epsilon>0$ be an arbitrarily small constant. We show that an $n_0$ exists such that $\forall n\geq n_0$, $ALG(n)/OPT(n)<9.25/6+\epsilon$. As mentioned in Section~\ref{sec:upperbounds}, for any even $n\geq 16$, $ALG(n)<9.25n+c$ for some small constant $c$.
In addition, by Theorem~\ref{thm:lowerbound}, for sufficiently large $n$, $OPT(n)>(6-\epsilon)n$. Thus, $ALG(n)/OPT(n)<(9.25n+c)/(n(6-\epsilon))$
Furthermore, for large enough $n$, $c/((6-\epsilon)n)<\epsilon/2$, and
\[\frac{ALG(n)}{OPT(n)}<\frac{9.25}{6-\epsilon}+\frac{\epsilon}{2}<\frac{9.25+3\epsilon}{6}+\frac{\epsilon}{2}=\frac{9.25}{6}+\epsilon\] \qed
\end{proof}

\subsubsection{Number of Crossings}
\begin{theorem}[Crossings lower bound]\label{thm:lowerboundcrossings}
Any knight's tour on an $n \times n$ board has at least $4n - O(1)$ crossings.
\end{theorem}

\begin{proof}
Let $T$ be an arbitrary knight's tour on an $n\times n$ board.
We show that $T$ has $n-O(1)$ crossings involving knight moves incident to the cells along the left edge of the board. An analogous argument holds for the three other edges of the board, which add up to the desired bound. 

We partition the edge cells along the left-most column into sets of three consecutive cells, which we call \textit{triplets} (if $n$ is not a multiple of three, we ignore any remaining cells, as they only account for a constant number of crossings). Two triplets are \textit{adjacent} if they contain adjacent cells. Each triplet has six associated knight moves in the tour $T$, two for each of its cells. We call the choice of moves the \textit{configuration} of the triplet. Since there are $\binom{4}{2} = 6$ choices of moves for each cell, there are $6^3 = 216$ possible configurations of each triplet (Figure~\ref{fig:tripletconstruction}, top).

\begin{figure}
	\centering
	\includegraphics[width=0.9\linewidth,page=27]{figures}
	\caption{\textbf{Top:} the 6 choices of moves for an edge cell. \textbf{Bottom left:} two of the $216$ possible triplet configurations. In the graph $G$ in the proof of Theorem~\ref{thm:lowerboundcrossings}, the nodes corresponding to the Triplet configurations $a$ and $b$ have weight 3 and 2, respectively, since the moves incident to their cells have that many crossings (the crossings are marked with white disks). \textbf{Bottom center:} illustration of the edge $a\rightarrow b$ in $G$. It has weight 5. \textbf{Bottom right:} example of a crossing (in red) between cells that are not in adjacent cell pairs.}
	\label{fig:tripletconstruction}
\end{figure}

Consider a weighted directed graph $G$ with a node for each of the $216$ possible triplet configurations and an edge from every node to every node, including a loop from each node to itself. The graph has weights on both vertices and edges. Given a node $v$, let $C(v)$ denote its associated configuration. The weight of $v$ is the number of crossings between moves in $C(v)$. The weight of each edge $v\rightarrow u$ is the number of crossings between moves in $C(v)$ and moves in $C(u)$ when $C(v)$ is adjacent and above $C(u)$ (Figure~\ref{fig:tripletconstruction}, bottom left and center).
Triplets have the property that if two knight moves in $T$ with endpoints in edge cells cross, the edge cells containing the endpoints are either in the same triplet or in adjacent triplets. This is not true for groups of two edge cells instead of three (Figure~\ref{fig:tripletconstruction}, bottom right).

Each path in $G$ represents a choice of move configurations for a sequence of consecutive triplets. Thus, the sum of the weights of the vertices and edges in the path equals the total number of crossings among all of these moves. Since $G$ is finite, any sufficiently long path eventually repeats a vertex. Given a cycle $c$, let $w(c)$ be the sum of the weights of nodes and edges in $c$, divided by the length of $c$. Let $c^*$ be the cycle in $G$ minimizing $w$. Then, $w(c^*)$ is a lower bound on the number of crossings per triplet along the edge.

The cycle minimizing $w$ can be found using Karp's algorithm for finding the minimum mean weight cycle in directed graphs~\cite{karp1978characterization}, which runs in $O(|V|\cdot|E|)$ time in a graph with vertex set $V$ and edge set $E$. However, this requires modifying the graph $G$, as Karp's algorithm is not suitable for graphs that also have node weights. We transform $G$ into a directed graph $G'$ with weights on only the edges and which preserves the optimal solution, as follows. We double each node $v$ in $G$ into two nodes $v_{in}, v_{out}$ in $G'$.  We also add an edge $v_{in} \rightarrow v_{out}$ in $G'$ with weight equal to the weight of $v$ in $G$. For each edge $v\rightarrow u$ in $G$, we add an edge $v_{out}\rightarrow u_{in}$ in $G'$.

Using Karp's algorithm, we found that, for $G$, $w(c^*)=3$. Figure~\ref{fig:mincrossings} shows an optimum cycle $c^*$, which in fact uses only one edge starting and ending at the same triplet configuration.
\end{proof}

\begin{figure}
	\centering
	\includegraphics[width=.26\linewidth,page=24]{figures}
	\caption{The configuration pattern along the board's edge with the minimum number of crossings. The dashed continuations illustrate that the moves in the configuration pattern can be extended to any number of columns without extra crossings.}
\label{fig:mincrossings}
\end{figure}

Since we only count crossings between moves incident to the first column, a question arises of whether the lower bound can be improved by considering configurations spanning more columns (e.g., the two or three leftmost columns). The answer is negative for any constant number of columns. Figure~\ref{fig:mincrossings} shows that the edges can be extended to paths that cover any fixed number of columns away from the edge without increasing the number of crossings.

Combined with the upper bound, we obtain the following approximation ratio.

\begin{corollary}
Algorithm~\ref{alg:main} achieves a $3+o(1)$ approximation on the minimum number of crossings.	
\end{corollary}

\section{Extensions}\label{a:extensions}

The idea of using formation moves to cover the board and junctions to close the tour is quite robust to variations of the problem. We show how this can be done in some of the most popular generalizations of the problem, formalized in the appropriate subsections: higher dimensions, boards with odd width and height and a missing corner, symmetric tours under 90 degree rotations, and tours for a different type of leaper, the $(1,4)$-leaper. In each case, there is a trivial lower bound on the number of turns of $\Omega(n)$, or $\Omega(n^{d-1})$ in the case of $d>2$ dimensions. This is because, in all these settings, the knight must turn after at most $n$ moves. In each case, our construction is within a constant factor of the minimum number of turns.

A variant that we do not consider are torus boards, where opposite edges are connected. This is because there already exists an algorithm that has $O(n)$ turns~\cite{watkins1997torus}, and crossings are not straightforward to define in torus boards. The problem of finding tours with a small number of turns seems easier on a torus board, because one is not forced to make a turn when reaching an edge. Nonetheless, in a square $n\times n$ torus board, $\Omega(n)$ turns are still required, because making $n$ consecutive moves in the same direction brings the knight back to the starting position, so at least one turn is required for each $n$ visited cells. The tours for torus boards by Watkins and Hoenigman~\cite{watkins1997torus} match this lower bound up to constant factors, at least for some board dimensions (see the last section in~\cite{watkins1997torus}).

\subsection{High-dimensional boards}\label{sec:3d}
We extend our technique to three and higher dimensions. In $d$ dimensions, a knight moves by $1$ and $2$ cells along any two dimensions, and leaves the remaining coordinates unchanged. A typical technique to extend a knight's tour algorithm to three dimensions is the ``layer-by-layer'' approach~\cite{watkins2012across,demaio2007chessboards,demaio2011prism}: a 2D tour is reused on each level of the 3D board, adding the minimal required modifications to connect them into a single tour.
We also follow this approach.
(Watkins and Yulan~\cite{yulan2006boxes} consider a generalizations of knight moves where the third coordinate also has a positive offset, but this is not as common.)

We require one dimenson of the board to be even and at least $16$ and another dimension to be at least $12$. We assume, without loss of generality, that these are the first two dimensions. The other dimensions can have any size. Note that at least one dimension must be even, or no tour exists~\cite{erde2012closed}.

For illustration purposes, we start with the 3D case, and later extend it to the general case.
The construction works as follows. The 2D construction from Algorithm~\ref{alg:main} is reused at each level of the third dimension. However, there are only two actual junctions, one on the first level, of height 5, and one on the last level, which may have any of the four heights in Figure~\ref{fig:junctions}. Every other bottom-left or top-right corner is replaced by a sequence of formation moves. At every level except the last, the formation ends adjacent to the corner using one of the sequences of moves in Figure~\ref{fig:3dcorners}. We show sequences for 4 different junction heights, guaranteeing that one shape fits for any sizes of the first two dimensions. The levels are connected with a formation move one level up and two cells to the side, as in Figure~\ref{fig:3djump}. At every level except the first, the formation starts with the rightmost sequence of moves in Figure~\ref{fig:3dcorners}.
A full example is illustrated in Figure~\ref{fig:262626tour}.

\begin{figure}
	\centering
	\includegraphics[width=0.99\linewidth,page=15]{figures}
	\caption{Corners where the knights stay in formation and end at specific positions.}
	\label{fig:3dcorners}
\end{figure}

\begin{figure}
	\centering
	\includegraphics[width=0.4\linewidth,page=16]{figures}
	\caption{Formation move across levels. Each color shows the starting and ending position of one of the knights.}
	\label{fig:3djump}
\end{figure}

\begin{figure}
	\centering
	\includegraphics[width=0.99\linewidth,page=17]{figures}
	\caption{Quartet moves for a $3D$ tour in a $26\times26\times26$ board. The quartet can move from the blue circle at each layer to the orange circle in the next layer by a quartet move.}
	\label{fig:262626tour}
\end{figure}

Since the sequence of moves between junctions is more involved than in two dimensions, Lemma~\ref{lem:seq} may not hold. There is, however, an easy fix: if the entire sequence is not a single cycle, replace the first junction with one that has a vertical matching (second junction in Figure~\ref{fig:junctions}, rotated 180 degrees). This then makes a cycle.

If the number of dimensions is higher than three, simply observe that the same move used between levels can also be used to jump to the next dimension; instead of changing by $1$ the third coordinate, change the fourth. After the first such move, the formation will be at the ``top'' of the second 3D board, which can be traversed downwards. This can be repeated any number of times, and generalizes to any number of dimensions.

In a $n^d$ board, $\Omega(n^{d-1})$ turns are needed, because there are $n^d$ cells and a turn must be made after at most $n/2$ moves. Our construction has $O(n^{d-1})$ turns, as it consists of $n^{d-2}$ iterations of the 2D tour. Thus, it achieves a constant approximation ratio on the minimum number of turns. We do not know of any lower bound on the number of crossings in higher dimensions. 

\subsection{Odd boards}\label{sec:odd}
We show how to construct a tour for a 2D board with odd dimensions which visits every cell except a corner cell. This is used in the next section to construct a tour that is symmetric under $90\degree$ rotations.

Let the board dimensions be $w\times h$, where $w>16$ and $h>12$ are both odd. First, we use Algorithm~\ref{alg:main} to construct a $(w-1)\times h$ tour which is missing the leftmost column. Then, we extend our tour to cover this column, except the bottom cell, with the variations of our construction depicted in Figure~\ref{fig:oddadaptation}. Since the sequence of moves between junctions is different than in the standard construction, Lemma~\ref{lem:seq} may not hold. Therefore, in the bottom-left junction, we might need to have a vertical or a horizontal matching. Figure~\ref{fig:oddadaptation} shows both options. For the top-left corner, recall that we use sequence $(3-h)\mbox{ mod }4$ in Figure~\ref{fig:othercorners}. Here, the height $h$ is odd, so we only need adaptations for Sequences 2 and 0.

Figure~\ref{fig:missingcornerconstruction} shows a complete tour with the adaptations from Figure~\ref{fig:oddadaptation}.

\begin{figure}
	\centering
	\includegraphics[width=0.9\linewidth,page=18]{figures}
	\caption{Adaptations required to add a row to the left of the standard construction, with a missing cell in the junction. \textbf{Left:} two options for the bottom-left junction, one with a vertical matching and one with a horizontal matching. \textbf{Center:} variations of Sequences 0 and 2 that reach one further column to the left. \textbf{Right:} variation for the transitions between diagonals along the left edge of the board that cover an extra column to the left.}
	\label{fig:oddadaptation}
\end{figure}

\begin{figure}
	\centering
	\includegraphics[width=0.7\linewidth,page=29]{figures}
	\caption{A $23\times 23$ tour with a missing corner.}
	\label{fig:missingcornerconstruction}
\end{figure}

\subsection{90 Degree Symmetry}\label{sec:symmetry}

In this section, we show how to construct a symmetric tour under 90 degree rotations.
We say a tour is symmetric under a given geometric operation if the tour looks the same when the operation is applied to the board.

As a side note, our construction is already nearly symmetric under $180\degree$ rotations. For board dimensions such as $30\times 30$ (Figure~\ref{fig:mainconstr}) where opposite corners have the same shape, the only asymmetry is in the internal wiring of the junctions. However, the construction cannot easily be made fully symmetric. It follows from the argument in the proof of Lemma~\ref{lem:seq} that if the two non-junction corners are equal, the entire sequence of formation moves from one junction to the other is neutral. Thus, using the same junction in both corners, as required to have symmetry, would result in two disjoint cycles.

Symmetric tours under $90\degree$ rotations exist only for square boards where the size $n=4k+2$ is a multiple of two but not a multiple of four~\cite{dejter1983symmetry}. 
Parberry~\cite{parberry1997efficient} gives a construction for knight's tours missing a corner cell and then shows how to combine four such tours into a single tour symmetric under $90\degree$ rotations. We follow the same approach  to obtain a symmetric tour with a number of turns and crossings linear on $n$, and thus constant approximation ratios.

In our construction from Section~\ref{sec:odd}, cell $(1,1)$ is missing, and edge  $e=\{(1,2),(3,1)\}$ is present. This suffices to construct a symmetric tour.
Divide the $2n\times 2n$ board into four equal quadrants, each of which is now a square board with odd dimensions. Use the construction for odd boards to fill each quadrant, rotated so that the missing cell is in the center. Finally, connect all four tours as in Figure~\ref{fig:symmetry}.

Figure~\ref{fig:symmetricconstruction} shows a complete tour obtained by connecting four rotated copies of the tour in Figure~\ref{fig:missingcornerconstruction} using the transformation from Figure~\ref{fig:symmetry}.

\begin{figure}
	\centering
	\includegraphics[width=0.9\linewidth,page=19]{figures}
	\caption{This transformation appears in~\cite{parberry1997efficient}. \textbf{Left:} four tours missing a corner square and containing a certain edge. The dashed lines represent the rest of the tour in each quadrant, which cover every square except the dark square. \textbf{Right:} single tour that is symmetric under $90\degree$ rotations. The numbers on the right side indicate the order in which each part of the tour is visited, showing that the tour is indeed a single cycle.}
	\label{fig:symmetry}
\end{figure}

\begin{figure}[t!]
	\centering
	\includegraphics[width=0.99\linewidth,page=30]{figures}
	\caption{A $46\times 46$ symmetric tour under 90 degree rotations.}
	\label{fig:symmetricconstruction}
\end{figure}

\subsection{Giraffe's tour}\label{sec:giraffe}

A giraffe is a leaper characterized by the move $(1,4)$ instead of $(1,2)$.
Giraffe's tours are known to exist on square boards of even size $n\geq 104$~\cite{kamcev2014generalised} and on square boards of size $2n$ when $n$ is odd and bigger than $4$~\cite{dejter1983symmetry}. Our result extends this to some rectangular sizes.

We adapt our techniques for finding giraffe's tours with $O(w+h)$ turns and crossings, where $w$ and $h$ are the width and height of the board.
We use a formation of $4\times 4$ giraffes.
Figure~\ref{fig:giraffemoves} shows the formation moves, Figure~\ref{fig:giraffeheel} shows the analogous of a heel to be used to transition between diagonals, and Figure~\ref{fig:giraffejunctions} shows the two junctions. Figure~\ref{fig:giraffetour} shows how these elements are combined to cover the board.

\begin{figure}
	\centering
	\includegraphics[width=0.6\linewidth,page=20]{figures}
	\caption{Formation of 16 giraffes moving together without leaving any unvisited squares.}
	\label{fig:giraffemoves}
\end{figure}

We restrict our construction to rectangular boards where $w=32k+20$, for some $k\geq 1$, and $h=8l+14$, for some $l\geq 1$. Extending the results to more boards would require additional heel variations. 

We start at the bottom-left junction as in the knight's tour. We transition between diagonals along the bottom edge with a giraffe heel, and along the top edge with a flipped giraffe heel. We transition between diagonals along the left and right edges with four consecutive upward moves. The junction has width 20 and each heel has width 32, so there are $k$ heels along the bottom and top edges. The junction has height 11 and the tip of the heel has height 3, so there are $l$ sequences of four upward moves along each side (e.g., Figure~\ref{fig:giraffetour}).

It is easy to see that the construction visits every cell. As in the case of knight's tours, for the result to be a valid giraffe's tour there must be a single cycle instead of a set of disjoint cycles. 
Note that the matchings in the two junctions in Figure~\ref{fig:giraffejunctions} form a cycle. Thus, if the formation reaches the top-right junction in the same matching as they left the bottom-left junction, the entire construction is a single cycle by an argument analogous to Theorem~\ref{thm:correctness}. Next, we argue that this is the case.

Let $H,F,$ and $U$ denote the sequences of formation moves in the heel, in the flipped heel, and the sequence of four upward moves, respectively. Let $T_{w,h}$ denote the entire sequence of moves from one junction to the other, where $w=32k+20$, for some $k\geq 1$, and $h=8l+14$, for some $l\geq 1$. Note that $T_{w,h}$ is a concatenation, in some order, of $H$ $k$ times, $F$ $k$ times, and $U$ $2l$ times (we can safely ignore diagonal moves, which do not change the coordinates of the giraffes within the formation). Let $M_1$ be the matching of the bottom-left junction. We want to argue that, after all the moves in $T_{w,h}$, the giraffes are still in matching $M_1$, that is, $T_{w,h}(M_1)=M_1$, using the notation from Section~\ref{sec:correctness}.

We show that not only the giraffes arrive to the other junction in the same matching but, in fact, they arrive in the same coordinates in the formation as they started. First, note that $U$ has the effect of flipping column $1$ with $2$ and column $3$ with $4$ in the formation. Perhaps surprisingly, $H$ and $F$ have the same effect. This is tedious but can be checked for each giraffe (Figure~\ref{fig:giraffeheel} shows one in red). Therefore, $T_{w,h}$ is equivalent to $U$ $2(l+k)$ times in a row. Note that after eight consecutive upward moves, or $U$ twice, each giraffe ends where it started. Thus, this is true of the entire tour.

\begin{figure}
	\centering
	\includegraphics[width=0.99\linewidth,page=21]{figures}
	\caption{A giraffe heel. The formation moves are shown with black arrows (grouping up to four sequential straight moves together). The intermediate positions of the formation are marked by rounded squares, showing that every cell is covered. Note that the tip of the heel fits tightly under the next heel. The red line shows the path of one specific giraffe in the formation.}
	\label{fig:giraffeheel}
\end{figure}

\begin{figure}
	\centering
	\includegraphics[width=0.99\linewidth,page=22]{figures}
	\caption{\textbf{Top:} two giraffe junctions. The bottom-left junction consists mostly of formation moves, whereas the top-right one was computed via brute-force search. \textbf{Bottom:} the matching $M_1$ corresponding to the bottom-left junction, the matching $M_2$ corresponding to the top-right junction, and the single cycle created by the union of these matchings. The cycle through the edges of the union is shown with the index of each node.}
	\label{fig:giraffejunctions}
\end{figure}

\begin{figure}
	\centering
	\includegraphics[width=0.99\linewidth,page=23]{figures}
	\caption{\textbf{Top:} the formation moves of a giraffe's tour on a $52\times 30$ board. \textbf{Bottom:} the corresponding giraffe tour.}
	\label{fig:giraffetour}
\end{figure}

\section{Conclusions}\label{sec:conclusions}
We have introduced two new metrics of ``simplicity'' for knight's tours: the number of turns and the number of crossings. We provided an algorithm which is within a constant factor of the minimum for both metrics. In doing so, we found that, in an $n\times n$ board, the minimum number of turns and crossings is $O(n)$.
Prior techniques such as divide-and-conquer necessarily result in $\Theta(n^2)$ turns and crossings, so
at the outset of this work it was unclear whether $o(n^2)$ could be achieved at all.

The ideas of the algorithm, while simple, seem to be new in the literature, which is interesting considering the long history of the problem. Perhaps it was our initial optimization goal that led us in a new direction.
The algorithm exhibits a number of positive traits. It is simple, efficient to compute, parallelizable, and amenable to generalizations (Section~\ref{a:extensions}). 
We conclude with some open questions:
\begin{itemize}
	\item Our tours have $9.25n+O(1)$ turns and $12n+O(1)$ crossings, and we showed respective lower bounds of $(6-\epsilon)n$ and $4n-O(1)$. The main open question is closing or reducing these gaps, as there may still be room for improvement in both directions. We conjecture that the minimum number of turns is at least $8n$.
	\item Are there other properties of knight's tours, besides turns and crossings, that might be interesting to optimize?
	\item Our method relies heavily on the topology of the ``knight-move'' graph. Thus, it is not applicable to general Hamiltonian cycle problems. Are there other graph classes with a similar enough structure that the ideas of formations and junctions can be useful?
\end{itemize}

\bibliography{biblio}
\appendix

\section{Simple Lower Bound for the Number of Turns}\label{a:easylb}
We show that any knight's tour on a $n\times n$ board has at least $4.25n-O(1)$ turns. Consider the cells in one of the two central columns (does not matter which one), and in a row in the range $(n/4,3n/4)$. They are shown in red in Figure~\ref{fig:badlowerbound}. These cells have the property that every maximal sequence of knight moves without turns through them reaches opposite facing edges. The maximal sequences of knight moves through the first and last red cells are shown in dashed lines. Because $n$ must be even, one of the two endpoints of each maximal sequence through a red cell is not an edge cell. It follows that each red cell is part of a sequence of knight's moves that ends in a turn at a non-edge cell. Thus, there is at least one turn at a non-edge cell for each pair of red cells. Since there are $0.5n$ red cells, we get the mentioned lower bound.

\begin{figure}[h!]
	\centering
	\includegraphics[width=0.4\linewidth,page=14]{figures}
	\caption{Setup for the $4.25n-O(1)$ lower bound on the number of turns.}
	\label{fig:badlowerbound}
\end{figure}

\end{document}